\documentclass[11pt]{article}

\usepackage[
letterpaper
,margin=2.9cm
]{geometry}

\usepackage{authblk}

\usepackage[utf8]{inputenc}
\usepackage{algo-viz}
\usepackage{amsthm}
\usepackage{xspace}
\usepackage{amsmath}
\usepackage{amssymb}
\newtheorem{theorem}{Theorem}
\newtheorem{definition}{Definition}
\newtheorem{lemma}{Lemma}
\usepackage{url}
\usepackage{caption,paralist}
\usepackage{subcaption}
\usepackage{float}
\usepackage{wrapfig}
\usepackage{thmtools,thm-restate}
\usepackage{amsfonts}

\usepackage[
colorlinks,
citecolor={red!80!black},
linkcolor={blue!60!black},
urlcolor={red!90!black}
]{hyperref}

\usetikzlibrary{matrix}
\usetikzlibrary{patterns}

\newcommand{\ie}{\emph{i.e.}}
\newcommand{\R}{\mathcal{R}}
\newcommand{\N}{\mathbb{N}}
\newcommand{\Z}{\mathbb{Z}}
\newcommand{\Views}{Views}
\newcommand{\tr}[1]{\vec{t}_{(#1)}}

\newcommand{\algo}[2]{\ensuremath{\mathcal{A}_{#1}^{\mathit{#2}}}\xspace}
\newcommand{\Cl}{\ensuremath{\mathit{Cl}}}
\newcommand{\IT}{\ensuremath{\mathit{IT}}}
\newcommand{\ITC}{\ensuremath{\mathcal{IT}}}

\newcommand{\Up}{\ensuremath{\mathit{Up}}}
\newcommand{\Down}{\ensuremath{\mathit{Down}}}
\newcommand{\Left}{\ensuremath{\mathit{Left}}}
\newcommand{\Right}{\ensuremath{\mathit{Right}}}
\newcommand{\Idle}{\ensuremath{\mathit{Idle}}}
\newcommand{\SRE}{\ensuremath{\mathit{SER}}\xspace}

\newcommand{\resetAlgoNodeStyle}{
\tikzstyle{algonode}=[draw=black, circle, semithick, minimum width=1.3em, minimum height=1.3em, inner sep=0.1em]
\tikzstyle{algonode-small}=[draw=black, circle, semithick, minimum width=0em, minimum height=0em, inner sep=0.2em]
}

\resetAlgoNodeStyle{}

\title{Infinite Grid Exploration by Disoriented Robots\thanks{This study has been partially supported by the \textsc{anr}
    projects \textsc{Descartes} (ANR-16-CE40-0023) and \textsc{Estate}
    (ANR-16-CE25-0009).}}
\author[1]{Quentin Bramas}
\author[2]{St\'ephane Devismes}
\author[3]{Pascal Lafourcade}
\affil[1]{University of Strasbourg, ICUBE}
\affil[2]{Universit\'e Grenoble Alpes, VERIMAG}
\affil[3]{University Clermont Auvergne, LIMOS}

\date{}

\begin{document}

\maketitle

\begin{abstract}
  We deal with a set of autonomous robots moving on an infinite
  grid. Those robots are opaque, have limited visibility capabilities,
  and run using synchronous Look-Compute-Move cycles. They all agree
  on a common chirality, but have no global compass.  Finally, they
  may use lights of different colors that can be seen by robots in
  their surroundings, but except from that, robots have neither
  persistent memories, nor communication mean.

  We consider the \emph{infinite grid exploration} (IGE) problem.
  For this problem we give two impossibility
  results and three algorithms, including one
  which is optimal in terms of number of robots.

  In more detail, we first show that two robots are not sufficient in
  our settings to solve the problem, even when robots have a common
  coordinate system. We then show that if the robots' coordinate
  systems are not self-consistent,
  three or four robots are not
  sufficient to solve the problem neither.

  Finally, we present three algorithms that solve the IGE problem in
  various settings.  The first algorithm uses six robots with constant
  colors and a visibility range of one. The second one uses the
  minimum number of robots, \ie, five, as well as five modifiable
  colors, still under visibility one. The last algorithm requires
  seven oblivious anonymous robots, yet assuming visibility
  two. Notice that the two last algorithms also  achieve
  exclusiveness.
\end{abstract}


\section{Introduction}

We deal with a swarm of mobile robots having low computation and
communication capabilities. The robots we consider are opaque ({\em
  i.e.}, a robot is able to see another robot if and only if no other
robot lies in the line segment joining them) and run in synchronous
Look-Compute-Move cycles, where they can sense their surroundings
within a limited visibility range. All robots agree on a common
chirality (\ie, when a robot is located on an axis of symmetry in its
surroundings, it is able to distinguish its two sides one from
another), but have no global compass (they agree neither on a
North-South, nor a East-West direction). However, they may use lights
of different
colors~\cite{10.1007/978-3-540-69507-3_5,Peleg:2005:DCA:2098351.2098353}. These
lights can be seen by robots in their surroundings.  However, except
from those lights, robots have neither persistent memories nor
communication capabilities.

We are interested in coordinating such weak robots, endowed with both
typically small visibility range ({\em i.e.}, one or two) and few
light colors (only a constant number of them), to solve an infinite
task in an infinite discrete environment.  As an attempt to tackle
this general problem, we consider the exploration of an infinite grid,
where nodes represent locations that can be sensed by robots and edges
represent the possibility for a robot to move from one location to
another.  Precisely, the exploration task consists in ensuring that
each node of the infinite grid is visited within finite time by at
least one robot. In the following, we refer to this problem as the
{\em Infinite Grid Exploration} (IGE) problem.

\paragraph{Contribution.} Our contribution consists of both negative
and positive results.  On the negative side, we show that if robots
have a common chirality but a bounded visibility range, then the IGE
problem is not solvable with
\begin{itemize}
\item two robots, even if those robots agree on common North (the
  proof of this result is essentially the adaptation to our context of
  the impossibility proof given
  in~\cite{Emek:2015:MAT:2852367.2852607});
\item three or four robots with self-inconsistent compass (\ie, the
  compass may change throughout the execution).
\end{itemize}
On the positive side, we provide three algorithms for solving the IGE
problem using opaque robots equipped with self-inconsistent compass,
yet agreeing on a common chirality.  Two of them additionally satisfy
{\em exclusiveness}~\cite{BBMR08j}, which requires any two robots to
never simultaneously occupy the same position nor traverse the same
edge.  The first one requires the minimum number of robots, \ie, five,
and ensures exclusiveness. The robots use modifiable lights with only
five states, and have a visibility range restricted to one. The second
algorithm solves the problem with six robots and only three
non-modifiable colors, still assuming visibility range one.  The last
algorithm requires seven identical robots without any light (\ie,
seven oblivious\footnote{{\em Oblivious} means that robots cannot
  remember the past.} anonymous robots) and ensures exclusiveness, yet
assuming visibility range two. Our contributions are summarized in the
table below.
\begin{center}
\begin{tabular}{|c|c|c|c|c|}
    \hline
    visibility range &  \# of robots & \# of colors & modifiable colors?     &  exclusiveness? \\\hline
    1 &  5 \textbf{(opt)}  & 5 & yes & yes \\\hline
    1 &  6  & 3 & no  & no\\\hline
    2 &  7 & 1 & N/A  & yes \\\hline
\end{tabular}
\end{center}
In order to help the reader, animations, for each of the three
algorithms, are available online~\cite{quentin_bramas_2019_2625730}.

\paragraph{Related Work.}

The model of robots with lights have been proposed by Peleg
in~\cite{10.1007/978-3-540-69507-3_5,Peleg:2005:DCA:2098351.2098353}.
In~\cite{Das:2016:AMR:2853249.2853724}, the authors use robots with
lights and compare the computational power of such robots with respect
to the three main execution model: fully-synchronous,
semi-synchronous, and asynchronous. Solutions for dedicated problems
such as {\em weak gathering} or {\em mutual visibility} have been
respectively investigated in~\cite{LFCPSV17j} and~\cite{OD18c}.

Mobile robot computing in infinite environments has been first studied in
the continuous two-dimensional Euclidean space.  In this context,
studied problems are mostly {\em terminating} tasks, such as \emph{pattern
  formation}~\cite{DP07j} and
\emph{gathering}~\cite{FLOCCHINI2005147}, \ie, problems where robots
aim at eventually stopping in a particular configuration specified by
their relative positions. A notable exception is the {\em flocking}
problem~\cite{YSDT11j}, {\em i.e.}, the infinite task consisting of
forming a desired pattern with the robots and make them moving
together while maintaining that formation.

When considering a discrete environment, space is defined as a graph,
where the nodes represent the possible locations that a robot can take
and the edges the possibility for a robot to move from one location to
another. In this setting, researchers have first considered finite
graphs and two variants of the exploration problem, respectively
called the {\em terminating} and {\em perpetual} exploration.  The
terminating exploration requires every possible location to be
eventually visited by at least one robot, with the additional
constraint that all robots stop moving after task completion. In
contrast, the perpetual exploration requires each location to be
visited infinitely often by all or a part of robots.
In~\cite{Devismes:2012:OGE:2413887.2413894}, authors solve terminating
exploration of any finite grid using few asynchronous anonymous
oblivious robots, yet assuming unbounded visibility range.
The exclusive perpetual exploration of a finite grid is considered in
the same model in~\cite{BMPT11c}.

Various terminating problems have been investigated in infinite grids
such as {\em arbitrary pattern formation}~\cite{BAKS19c}, {\em mutual
  visibility}~\cite{ABKS18c}, and {\em gathering}~\cite{DN16,DDC17c}.
The possibly closest related work is that of Emek {\em et
  al.}~\cite{Emek:2015:MAT:2852367.2852607}.  In this paper, authors
consider a treasure search problem, which is roughly equivalent to the
IGE problem, in an infinite grid. They consider robots that operate in
two models: the semi-synchronous and synchronous ones.  However, they
do not impose the exclusivity at all since their robots can only sense
the states of the robots located at the same node (in that sense, the
visibility range is zero). The main difference with our settings is
that they assume all robots agree on a {\em global compass}, {\em
  i.e.}, they all agree on the same directions North-South and
East-West; while we only assume here a {\em common chirality}.  This
difference makes their model stronger, indeed they propose two
algorithms that respectively need three synchronous and four
asynchronous robots, while in our settings the IGE problem (even in
its non-exclusive variant) requires at least five robots. They also
exclude solutions for two robots.

In a followup paper~\cite{brandt2018tight}, Brandt {\em et al.}
 extend the impossibility result of Emek {\em et al.} Indeed,
they show the impossibility of exploring an infinite grid with three
semi-synchronous deterministic robots that agree on a common
coordinate system. Although proven using similar techniques, this
result is not correlated to ours. Indeed, the lower bound of Brandt
{\em et al.} holds for robots that are weaker in terms of synchrony
assumption (semi-synchronous {\em vs.} fully synchronous in our case),
but stronger in terms of coordination capabilities (common coordinate
system {\em vs.} self-inconsistent compass in our case).  In
other words, our impossibility results does not (even indirectly)
follows from those of Brandt {\em et al.} since in our model
difficulties arise from the lack of coordination capabilities and not
the level asynchrony. As a matter of facts, based on the results of
Emek {\em et al.}~\cite{Emek:2015:MAT:2852367.2852607}, four
(asynchronous) robots are actually necessary and sufficient in their
settings,  while it is five in our context.

\paragraph{Roadmap.}
In the next section, we define our computational model. In
Section~\ref{sec:impossible}, we present lower bounds on the number of
robots to solve the IGE problem. In Sections~\ref{sec:algos} and
Section~\ref{sec:visibility2}, we propose algorithms solving the
IGE under visibility range one and two.

\section{Model}\label{sec:model}

We consider a set $\R$ of $n > 0$ robots located on an \emph{infinite
  grid} graph with vertex set in $\Z\times\Z$, \ie, there is an edge
between two nodes $(i,j)$ and $(k,l)$ if and only if the
\emph{Manhattan distance} between those two nodes, {\em i.e.},
$|i-k|+|j-l|$, is 1.  The coordinates of the nodes are used for the
analysis only, \ie, robots cannot access them.

We assume time is discrete and at each \emph{round}, the robots
synchronously perform a \emph{Look-Compute-Move} cycle.  In the
\emph{Look} phase, a robot gets a snapshot of the subgraph induced by the
nodes at distance $\Phi > 0$ from its position. $\Phi$ is called the
{\em visibility range} of the robots.
The snapshot is not oriented in any way as the robots do not agree on
a common North.  However, it is implicitly ego-centered since the
robot that performs a Look phase is located at the center of the
subgraph in the obtained snapshot.
Then, each robot \emph{computes} a destination (either
Up, Left, Down, Right or Idle) based only on the snapshot it received.
Finally, it \emph{moves} towards its computed destination. 
 We also assume that robots are
\emph{opaque} and can obstruct the visibility so that if three robots
are aligned, the two extremities cannot see each other.

Robots may have \emph{Lights} with different colors that can be seen
by robots within distance $\Phi$ from them. Let $\Cl$ be the set of possible
colors.  Even when an algorithm does not achieve exclusiveness, we
forbid any two robots to occupy the same node simultaneously. So, the
\emph{state} of a node is either the color of the light of the robot
located at this node, if there is one, or $\bot$ otherwise. If there
is a robot we say the node is occupied, otherwise we say it is empty.

In the Look phase, the snapshot includes the state of the nodes (at
distance $\Phi$). After the compute phase, and if colors are
\emph{modifiable}, a robot may decide to change the color of its
light. Otherwise, colors are said to be \emph{fixed}.

\paragraph{Configurations.}
A \emph{configuration} $C$ is a set of couples $(p,c)$ where $p\in
\Z\times\Z$ is an occupied node and $c\in \Cl$ is the light's color of the
robot located at $p$. A node $p$ is empty if and only if $\forall
c,\,(p, c)\notin C$. We sometimes just write the set of occupied node
when the colors are clear from the context. For better readability, we
sometimes partition the configuration into several subsets $C_1, \ldots,
C_k$ and write $C=\{C_1, \ldots, C_k\}$ instead of writing $(C =
C_1\cup\ldots\cup C_k) \land (\forall i\neq j, C_i\cap
C_j=\emptyset)$.

\paragraph{Views.}

We denote by $G_r$ the {\em globally oriented view} centered at the
robot $r$, \ie, the subset of the configuration containing the states
of the nodes at distance at most $\Phi$ from $r$, translated such that
the coordinates of $r$ is $(0,0)$. We use this globally oriented view
in our analysis to describe the movements of the robots: when we say
"the robot moves Up", it is according to the globally oriented view.
However, since robots do not agree on a common North, they have no
access to the globally oriented view.  When a robot looks at its
surroundings, it obtains a snapshot.
To model
this, we assume that, the {\em local view} acquired by a robot $r$ in
the Look phase is the result of an arbitrary combination of
\emph{indistinguishable transformations} on $G_r$.  The set $\IT$ of
indistinguishable transformations depends on the assumptions we make
on the robots. The rotations of angle $\pi/2$, $\pi$ and $3\pi/2$,
centered at $r$ are in $\IT$ if and only if the robots do not agree on
a common North direction. A mirroring is in $\IT$ if and only if the
robots do not agree on a common \emph{chirality} (they cannot
distinguish between clockwise and counterclockwise).
Moreover, in the obstructed visibility model, the function that
removes the state of a node $u$ if there is another robot between $u$
and $r$ is in $\IT$ and is systematically applied.
$\ITC$ denotes the set of possible combinations of indistinguishable transformations.

For a robot $r$, if the same transformation $f_r\in\ITC$ is used for every
look phase of $r$, we say that $r$ is
\emph{self-consistent}. Otherwise, an adversary can choose a different
transformation for each look phase, and $r$ is said to be \emph{self-inconsistent}.

In the rest of the paper, all our algorithms assume that all 
robots agree on a common chirality, \ie, they can distinguish two
mirrored views, but we make no assumption on the self-consistency of
the coordinate system. On the other hand, we give impossibility
results for stronger model when possible.

When a robot $r$ computes a destination $d$, it is relative to its
local view $f(G_r)$, which is the globally oriented view transformed by $f\in\ITC$. It
is important to see that the actual movement of the robot in the
\emph{globally oriented view} is actually $f^{-1}(d)$. Indeed, if $d=\Up$ but the
robot sees the grid upside-down ($f$ is the $\pi$-rotation), then the
robot moves $\Down = f^{-1}(\Up)$. %
In a configuration $C$, $V_C(i,j)$ denotes the globally oriented view of a robot
located at $(i,j)$.

\paragraph{Exploration Algorithm.}

An algorithm \algo{}{} is a tuple $(\Cl, I, T)$ where $\Cl$ is the set of
possible colors, $I$ is the initial configuration, and
$T$ is the transition function $\Views\rightarrow\{\Idle, \Up, \Left,
\Down$, $\Right\}\times \Cl$, where $\Views$ is the set of possible globally oriented views.

Recall that we assume in our algorithms that the robots are not
self-consistent. In this context, we say that an algorithm $(\Cl, I,
T)$ is \emph{well-defined} if the global destination computed by a
robot does not depend on the transformation $f$ chosen by the
adversary, \ie, for every globally oriented view $V$, and every
transformation $f\in\ITC$, we have $T(V) = f^{-1}(T(f(V)))$.  This is
usually a property obtained by construction of the algorithm, as we
describe the destination $d$ for a given globally oriented view $V$
and then assume that the destination computed from local view $f(V)$
is $f(d)$, for any $f\in \ITC$.

We can extend the transition function $T$ to the entire configuration. When the robots are in configuration $C$, the configuration obtained after one round of execution is denoted $T(C)$ and contains the couple $((i,j), c)$ if and only one of the following condition is verified:

\begin{itemize}
    \item $(i,j) \in C$ and $T(V_C(i,j)) = (Idle, c)$,
    \item $(i-1,j) \in C$ is occupied and $T(V_C(i-1,j)) = (Right, c)$,
    \item $(i+1,j) \in C$ is occupied and $T(V_C(i+1,j)) = (Left, c)$,
    \item $(i,j-1) \in C$ is occupied and $T(V_C(i,j-1)) = (Up, c)$,
    \item $(i,j+1) \in C$ is occupied and $T(V_C(i,j+1)) = (Down, c)$.
\end{itemize}

In the remaining of the paper, we sometime write $\algo{}{}(C)$ instead of $T(C)$.
The execution of an algorithm is the sequence $(C_i)_{i\in\N}$ of configurations, such that $C_0 = I$ and \hbox{$\forall i \geq 0$, $C_{i+1} = T(C_{i})$}.

\begin{definition}[Infinite Grid Exploration]
An algorithm \algo{}{} solves the \emph{infinite grid
  exploration} (IGE) problem if in the execution $(C_i)_{i\in\N}$ of \algo{}{} and
for any node $(i,j)\in \Z\times\Z$ of the grid, there exists $t \in \N$ such that $(i,j)$ is occupied in $C_t$.
\end{definition}

\paragraph{Notations.}

$\tr{i,j}(C)$ denotes the translation of the configuration $C$ of
vector $(i,j)$.

\section{Impossibility Results}\label{sec:impossible}

The lemma below states the intuitive, yet non trivial, idea that, in
order to explore an infinite grid, the maximum distance between two
farthest robots should tend to infinity. This claim is the cornerstone in the proofs of our impossibility results.

\begin{lemma}\label{lem:distance is increasing}
    Let $(C_i)_{i\in\N}$ be an execution of a given algorithm \algo{}{}.
    Let $d_i$ be the distance between the two farthest robots in
    Configuration $C_i$.  If \algo{}{} solves the IGE problem, then
    $\lim_{i\rightarrow+\infty} d_i = +\infty$.
\end{lemma}
\begin{proof}
    We proceed by the contradiction. So we suppose there exists a bound $B > 0$
    such that there are infinitely many configurations where the distance
    between every pair of robots is less than $B$. In other words,
    there is a subsequence of $(C_i)_{i\in\N}$ where the distance
    between every pair of robots is less than $B$.  Let $(b_i)_{i\in
      \N}$ be the sequence of indices of this subsequence, \ie,
    $(b_i)_{i\in \N}$ is a strictly increasing sequence of integers
    such that $d_{b_i} < B$.

    When all robots are at distance less than $B$, then the occupied
    positions are included in a square sub-grid of size $B\times B$.
    Since the number of possible configurations included in a sub-grid
    of size $B\times B$ is finite, there must be two indices $k$ and
    $l$ such that $C_{b_l} = t(C_{b_k})$ and $k < l$ for a given
    translation $t$. The movements done by the robots in
    Configurations $C_{b_k}$ and $C_{b_l}$ are the same because each
    robot has the same globally oriented view in both configurations,
    only their positions change. Thus $C_{b_l + 1} = t(C_{b_k + 1})$
    and so on so forth, so that $\forall i,\, C_{b_l + i} = t(C_{b_k +
      i})$. We obtain that the configurations are periodic (with
    period $P=b_l-b_k$) and a node $u$ is visited if and only if it is
    visited before round $b_l$ or if there exists a node $v$ visited
    between round $b_k$ and $b_l$ such that $u = t^q(v)$ with
    $q>0$. So, we claim that there exists a node that is never
    visited.

    To prove this claim, we now exhibit such a node.  Let $I$ be the
    set of integers $i$ such that $(t^{-1})^i(0,0)$ is visited before round
    $b_l$
    applied $i$ times. $I$ is finite because the number of nodes
    visited before $b_l$ is finite.  Let $m$ be the maximum integer in
    $I$ (or 0 if $I$ is empty). Let $u = (t^{-1})^{m+1}(0,0)$. Then,
    clearly $u$ is not visited before round $b_l$, otherwise we have a
    contradiction with the maximality of $m$. Moreover, $u$ cannot be
    visited after round $b_l$, otherwise $u$ would be equals to
    $t^q(v)$ for a given integer $q$ and a given node $v$, visited
    between round $b_k$ and $b_l$, \ie, $v = (t^{-1})^q(u) =
    (t^{-1})^{q+m+1}(0,0)$, which also contradicts the maximality of
    $m$. Thus $u$ is never visited.
\end{proof}

The next theorem shows the impossibility of exploring the infinite
grid with two robots, even if they agree on a common coordinate
system. The proof of the theorem is essentially an adaptation to our
context of the proof of impossibility given in
\cite{Emek:2015:MAT:2852367.2852607}; yet it is necessary since our
model is not comparable to that
of~\cite{Emek:2015:MAT:2852367.2852607}.

\begin{theorem}\label{theo:impossibility with 2 robots}
    No algorithm can solve the IGE problems using two robots, even if
    robots agree on common North and chirality.
\end{theorem}
\begin{proof}
  By the previous Lemma, there is a configuration from which the two
  robots will no more see each other (their distance will remain
  greater than an arbitrary bound $B\geq \Phi$). For each robot, its
  next move only depends on the color of its light. Since the number
  of color is finite, the movements of each robot are periodic and
  using the same argument as in the proof of the previous lemma, we
  can conclude that a node is never visited.
\end{proof}

\begin{lemma}\label{lem:a robot alone stay idle}
  A robot with self-inconsistent compass and that sees no other robot,
  either stays idle or the adversary can make it alternatively moving
  between two chosen adjacent nodes.
\end{lemma}
\begin{proof}
    If such a robot do not stay idle, it moves toward a direction
    $d\in\{Up, Down, Left, Right\}$ but since its orientation is not
    self-consistent, the adversary can chose, for each activation, a
    transformation $f\in\ITC$ such that the destination $f^{-1}(d)$ in the
    globally oriented view alternate between two chosen directions ({\em e.g.}, $Up$
    and $Down$).
\end{proof}

\begin{theorem}\label{theo:imp:3}
    It is impossible to solve the IGE problem with three robots equipped
with self-inconsistent compasses that agree on a common chirality.
\end{theorem}
\begin{proof}
  By Lemma~\ref{lem:distance is increasing}, there is a configuration
  where two robots are always at distance at least $B$ (say $B >
  2\cdot\Phi+2$), so that it is impossible for a robot to see the two
  other robots in the same snapshot. Since there are three robots, at
  least one robot $r$ does not see any other robot. By
  Lemma~\ref{lem:a robot alone stay idle}, if $r$ stays alone, then it
  remains idle or the adversary can make it alternatively moves
  between two nodes infinitely often. Moreover, the two other robots
  cannot explore the grid alone, by Theorem~\ref{theo:impossibility
    with 2 robots}. Now, they cannot both move towards $r$ because in
  such a case the distance between the farthest robots would becomes
  less than $B$, a contradiction.  Finally, if one of the two other
  robots moves towards $r$, at some point all robots are out of the
  visibility range of each other. In that case, the adversary can make
  the exploration fail, by Lemma~\ref{lem:a robot alone stay idle}.
\end{proof}

\begin{theorem}\label{theoremetrois}
    It is impossible to solve the IGE problem with four robots equipped with self-inconsistent compasses that agree on a common chirality.
\end{theorem}
\begin{proof}
  Assume, by contradiction, that an algorithm $\mathcal{A}$ solves the
  IGE problem with four robots equipped with self-inconsistent
  compasses that agree on a common chirality.  The outline of the
  proof is as follows. We first prove that, after some time, a moving
  group travels infinitely often between two robots by 
  periodically performing the same translation. This implies that,
  after some time, the movements of the robots depend only on
  configurations of bounded size. From this latter statement, we will
  deduce that the movements of the two farthest robots are periodic,
  leading to a contradiction.

    Like in the previous impossibility results, we consider a round
    where two robots are at distance $B \gg \Phi$, by
    Lemma~\ref{lem:distance is increasing}. The two farthest robots
    are called the {\em extremities} of the configuration. In this
    proof we denote $a$ and $b$ the two extremities, and write
    abusively \emph{the robots at extremity $a$,} (resp. $b$). Using
    the same argument as in Theorem~\ref{theo:imp:3}, we know that at
    some round $t_0$, three robots are located at one extremity and
    the fourth robot is waiting on the other extremity. Since one
    robot cannot travel an arbitrary distance, then, at some round
    $t\geq t_0$, two robots leave one extremity to move towards the
    second one. Two robots traveling from one extremity to the other
    happens infinitely often, otherwise, one extremity remains
    isolated forever and the three other robots solves the IGE
    problems, which is impossible by Theorem~\ref{theo:imp:3}. A group
    of two robots that is moving from one extremity to the other is
    called a \emph{moving group}.

    When two robots are moving, their movements depend on their lights' colors and their relative positions. Since there is a finite number of such combinations (the number of colors and the visibility range are finite), after a finite number of rounds, the sequence of movements performed by the moving group is periodic. During one period, the group performs a translation. After some rounds, the distance between the extremities is such that a moving group traveling from one extremity to the other uses at least one periodic sequence of movements. Again, the number of such periodic sequences of movements, and the number of possible translations, is finite.

    Let $\{t_{1},\ldots, t_{k} \}$ the set of translations that are used infinitely often by a moving group to travel from one extremity to the other. 
\begin{description}
\item[Claim 1:] {\em all $t_i$ are collinear.}

{\em Proof of the Claim:}
    Take two translations $t_1$ and $t_2$ such that, infinitely often \emph{(i)} $t_1$ is used by a moving group to travel from extremity $a$ to extremity $b$, and \emph{(ii)} $t_2$ is used for the return journey to extremity $a$. Let $v_1$, resp. $v_2$, the vector associated with translation $t_1$, resp. $t_2$.
    Fact \emph{(i)}, resp. \emph{(ii)}, implies that there is a bound $d$ such that $a$ and $b$ are at distance at most $d$ from a line directed by vector $v_1$, resp. $v_2$. By lemma~\ref{lem:distance is increasing}, we have a round where the distance between $a$ and $b$ is arbitrary large, thus the angle between the line passing through $a$ and $b$ and any line directed by $v_1$, resp. $v_2$, is smaller than an arbitrary number $\varepsilon>0$. Thus, the angle between a line directed by $v_1$ and a line directed by $v_2$ is null, in other word $v_1$ and $v_2$ are collinear.
\end{description}
Since all $t_i$ are collinear and translations are defined on the discrete grid, we can define the translation $t$ as the \emph{least common multiple} of $t_1,\ldots,t_k$. In more detail, $t$ is such that for all $1\leq i\leq k$, there exists $l_i\in\Z$ verifying $t=t_i^{l_i}$ where $t_i^{l_i}$ is the composition of $t_i$ with itself $l_i$ times. Thus, we can assume that only $t$ and $t^{-1}$ are used infinitely often ($t$ for the outward journey from extremity $a$ to $b$ and $t^{-1}$ for the return journey). 

There might be several different sequences of movements that translate a moving group by translation $t$, each of them having different length (\ie, number of rounds). But, without loss of generality (by taking $t^{|\Cl|!}$ instead of $t$ for instance), we can assume that the length of all such sequences of movements is a multiple of $|\Cl|!$. 
By the following claim, the light's color and the position of a robot that remains alone are unchanged after a moving group performs one or more than one translation.

\begin{description}
\item[Claim 2:] {\em If a robot remains alone during $2M$ rounds, such that $|\Cl|!$ divides $M$, then, its light's color and its position after $M$ and after $2M$ are the same.
}

{\em Proof of the Claim:}
A robot that is alone can change its color, but the sequence of colors is periodic with period at most $|\Cl|$ so every $M$ rounds, since $|\Cl|!$ divides $M$, its light reach back the same color.
Moreover, by Lemma~\ref{lem:a robot alone stay idle}, either Algorithm \algo{}{} always decides that the robot moves, and the adversary can enforce the robot to be at the same location after $M$ rounds ($M$ is a multiple of two), or \algo{}{} decides that the robot moves only a finite number of rounds (which must be at most the number of colors) so after $M$ rounds it stays idle and its position no more changes during the next $M$ rounds.
\end{description}

    We choose the bound $B$ such that every time the moving group travels from one extremity to the other, it performs at least two times a periodic sequence of movements corresponding to translation $t$ or $t^{-1}$. By the previous claim, we know that when a moving group of robots reaches an extremity, the colors of their lights and their relative positions do not depend on the number of times the translation $t$ or $t^{-1}$ was performed. So, every time the moving group reaches an extremity, it does not know how much distance it has traveled.




    We now define a sequence of configurations that captures all the movements of extremity $a$ in the execution of \algo{}{}.
    Let $C_0=\{M_0, a_0,b_0\}$ be the first configuration
from which the two farthest robots are always at distance at least $B$ and 
such that there are a robot $a_0$ at extremity $a$, a robot $b_0$ at extremity $b$, and a group $M_0$ of two robots at bounded distance from $a_0$ (if there are several choices for the extremity $a_0$, we can choose any of them). In such a configuration, we say that {\em there are three robots at ''extremity $a$''}.
    Similarly, let $C_i=\{M_i, a_i,b_i\}$ be the $i^{\text{th}}$ configuration where there are three robots $M_i\cup \{a_i\}$ at the extremity $a$.
    The sequence $(C_i)_{i\in \N}$ captures all the configurations of the execution of \algo{}{} where three robots are located at extremity $a$, hence all the configurations impacting the movement of extremity $a$. Between configurations $C_i$ and $C_{i+1}$, one or more rounds of execution may happen, and we are particularly interested by the case when between $C_i$ and $C_{i+1}$ a moving group travels to extremity $b$ and come back to extremity $a$, which we know happens infinitely often. In this case, the robot left alone at extremity $a$ either stays idle or may be forced to move to an adjacent node by Lemma~\ref{lem:a robot alone stay idle}, but that is uniquely determined by configuration $C_{i}$.

    Fix an index $i$ such that between configurations $C_i$ and $C_{i+1}$ a moving group travels towards $b$ and come back to $a$. When the moving group $M_i$ moves towards $b_i$, it performs $k_i$ times a periodic sequence of movements associated with the translation $t$. Then it reaches $b_i$, and after some more rounds, another group $M_{i}'$ leaves the extremity $b$ and reaches back the extremity $a$. For the return journey, the moving group performs $k_i'$ times the translation $t^{-1}$. As soon as the moving group $M_{i'}$ reaches the extremity $a$, the configuration is $C_{i+1}$.

    Now we prove that the movement of the extremity $a$ is actually determined by a recursive sequence of bounded-size configurations, which implies that the extremity $a$ performs a periodic movement.

    To exhibit such sequence, we observe that the round-trip to extremity $b$ of a moving group would be similar if the extremity $b$ was translated by $t^{-1}$, \ie, if $b$ was closer to extremity $a$ so that one less translation is performed by the moving group to reach extremity $b$. Indeed, from the previous observation, when a moving group reaches extremity $b$, their light's color and relative positions are the same whether the moving robots performed one or more than one translation.
    
    Hence, if the round-trip, between configurations $C_i$ and $C_{i+1}$, uses $k_i$ times the translation $t$ for the outward journey and $k_i'$ time the translation $t^{-1}$ for the return journey, then the round-trip would be similar if $b_i$ was translated by $t^{-k_i^{\min}}$, with $k_i^{\min}=\min(k_i, k_i')-1$. If no round-trip is performed between configurations $C_i$ and $C_{i+1}$, then we can define $k_i^{\min}$ as the minimal number of times we can translate $b_i$ by $t^{-1}$ such that the configuration is contained in a sub-grid of size bounded by $B\times B$. In this case, again, the execution of the algorithm would be similar whether or not $b_i$ is closer to $a_i$. 

    Formally, we define the \emph{compressed} version of $C_i$ denoted $D_i$ as follow: $D_i = \{M_i, a_i, t^{-k_i^{\min}}(b_i)\}$, with $k_i^{\min}=\min(k_i,k_i')-1$ if a moving group does a round-trip to extremity $b_i$ between $C_i$ and $C_{i+1}$. Otherwise, $k_i^{\min}$ is the minimal number of times we can translate $b_i$ so that $D_i$ is contained in a sub-grid of size bounded $B\times B$.

    \newcommand{\kmin}{\mathbf{k}} One can see that, in all cases, the
    configuration $D_i$ is included in a sub-grid of size $B\times
    B$. Indeed, in the first case, it is possible either to travel
    from $a_i$ to $t^{-k_i^{\min}}(b_i)$ using the translation $t$
    only once, or from $t^{-k_i^{\min}}(b_i)$ to $a_i$ using the
    translation $t^{-1}$ only once (either $k_i^{\min} = k_i - 1$ or
    $k_i^{\min} = k_i' - 1$), and by definition of the bound $B$, a
    group cannot travel a distance greater than $B$ using only once
    translation $t$ or $t^{-1}$.  For better readability, let simply
    denote $k_i^{\min}$ by $\kmin_i$.

    We denote by $g$ the function that takes a configuration $\{M, r_a, r_b\}$, where $M\cup \{r_a\}$ are at extremity $a$, and return the configuration $\{M,r_a,t^{-1}(r_b)\}$. $g$ is well-defined for configurations where $t^{-1}(r_b)$ does not intersect with $M\cup \{r_a\}$. Also, we see here that the function $g$ does not depend on the choice of the robot $r_a$ at the extremity $a$. We now have $D_i = g^{\kmin_i}(C_i)$, where $g^{\kmin_i}$ is the composition of $g$ with itself $\kmin_i$ times. We can see in Fig.~\ref{fig:config diagram} the relations between the defined configurations.

  \begin{figure}\centering
    \begin{tikzpicture}
      \matrix (m) [matrix of math nodes,row sep=3em,column sep=4em,minimum width=2em]
      {
         C_i & C_{i+1} &  \\
         D_i & D_{i}' & D_{i+1}\\};
      \path[-stealth]
        (m-1-1) edge node [left] {$g^{\kmin_i}$} (m-2-1)
                edge [] node [below] {$\mathcal{A}^+$} (m-1-2)
        (m-2-1.east|-m-2-3) edge [] node [below] {$\mathcal{A}^+$} (m-2-2)
        (m-2-2) edge [] node [right] {$g^{-\kmin_i}$} (m-1-2)
        (m-1-2) edge node [above right] {$g^{k_{i+1}}$} (m-2-3)
        (m-2-2) edge node [below] {$h$} (m-2-3);
    \end{tikzpicture}
    \caption{Relations the different defined configurations. $\mathcal{A}^+$ corresponds to one or more execution of Algorithm $\mathcal{A}$}\label{fig:config diagram}
\end{figure}
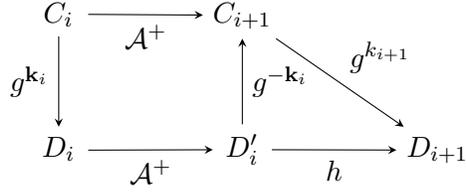

    If we apply the algorithm from a configuration $D_i$ the moving group starts with the same movements as in $C_i$ (since $M_i$ and $a_i$ are unchanged) and reaches the extremity $t^{-\kmin_i}(b_i)$ by
    performing $\kmin_i$ less times the translation $t$, but at least once, so it arrives in the same state as when the moving group reaches $b_i$ starting from configuration $C_i$. So the movement performed by the three robots at this extremity $b$ will be the same as when starting from configuration $C_i$. The return journey to extremity $a$ is also the same but
    performing $\kmin_i$ less times the translation $t^{-1}$
    compared to the return journey when starting with configuration
    $C_i$. When the moving group reaches back the extremity $a$,
    it is impossible for the robots to differentiate between the two
    executions, the one starting with $C_i$ and the other starting with
    $D_i$, the two obtained configurations being $C_{i+1}$ and $D_{i}'=\{M_{i+1}, a_{i+1},
    t^{-\kmin_i}(b_{i+1})\}$.

    We now prove that there exists a function $h$ that maps a configuration $D_{i}'$ to $D_{i+1}$, $\forall i\geq 0$. To show that $h$ is well-defined we have to show that, $\forall i,j\geq 0$ such that $D_i' = D_j'$, then $D_{i+1} = D_{j+1}$. Indeed, assume without loss of generality that $\kmin_j \geq \kmin_i$, then we have
    \[
        C_{i+1} = g^{-\kmin_{i}}(D_i') = g^{-\kmin_{i}}(D_j') = g^{-\kmin_{i}}(g^{\kmin_j}(C_{j+1}))= g^{\kmin_j-\kmin_{i}}(C_{j+1})
    \]
    But this implies that $C_{i+1}$ and $C_{j+1}$ have the same compressed version, \ie, $D_{i+1} = D_{j+1}$.

    Thus, we have that $D_i = h(\mathcal{A}^+(D_{i-1}))$, where $\mathcal{A}^+$ corresponds to one or more executions of Algorithm $\mathcal{A}$, and since the configuration $D_i$ is included in a subgrid of size $B\times B$, and there are a finite number of such configurations up to a translation, then there exists a translation $t_D$ and two indices $q > p$ such that $D_q = t_D(D_p)$. Since the globally oriented view of the extremity $a_i$ is the same in configuration $C_i$ and in the compressed configuration $D_i$, we deduce that the movements of the extremity $a$ are periodic.
    By a symmetric argument, the same is true for the other extremity, so that each extremity is moving in one direction.
    Since the moving group is always at bounded distance from the segment delimited by the two extremities, we have that, \emph{(i)} if the direction of the movements of the extremities are collinear, then no robots can move arbitrary far away from a given line (see Fig.~\ref{fig:collinear extremities}), \emph{(ii)} otherwise, no robots can move arbitrary far away from a given cone (see Fig.~\ref{fig:non collinear extremities}). In both case, not all nodes are visited.

\tikzstyle{robot}=[fill, circle, inner sep=0.04cm]
\tikzstyle{robotm}=[fill, rectangle, inner sep=0.08cm]

\begin{figure}[H]
\begin{minipage}{0.4\textwidth}\centering
    \begin{tikzpicture}
    \fill[pattern color=black!20, pattern=north east lines] (-1,-1) -- (3,-1) -- (3,1) -- (-1, 1) -- cycle;
    \draw[color=black!40] (-1,-1) -- (3,-1);
    \draw[color=black!40] (-1,1) -- (3,1);

    \node[robot,label={90:$a$}] at (0,0) {} edge[->] (-0.6,0);
    \node[robot,label={90:$b$}] at (2,0) {} edge[->] (2.6,0);

    \node[robotm,label={90:$M$}] at (1,0) {};
    \draw (1,0) -- (1.8,0) -- (1.8,0.15) edge[->] (1.6,0.15);

    \node[] at (1,-0.7) {\small visited};
    \node[] at (1,-1.3) {\small not visited};
\end{tikzpicture}
\captionof{figure}{visited nodes if the direction of the extremities are collinear}\label{fig:collinear extremities}
\end{minipage}~~~~~~~~~~%
\begin{minipage}{0.6\textwidth}\centering
    \begin{tikzpicture}
    \fill[pattern color=black!20, pattern=north east lines] (-2,) -- (-0.5,-1) -- (2.5,-1) -- (4, 1) -- cycle;
    \draw[color=black!40] (-0.5,-1) -- (-2,1);
    \draw[color=black!40] (2.5,-1) -- (4,1);
    \draw[color=black!40] (-0.5,-1) -- (2.5,-1);

    \node[robotm,label={90:$M$}] at (1,0) {};
    \draw (1,0) -- (1.8,0) -- (1.9,0.15) edge[->] (1.6,0.15);

    \node[robot,label={90:$a$}] at (0,0) {} edge[->] (-0.5,0.5);
    \node[robot,label={90:$b$}] at (2,0) {} edge[->] (2.5,0.5);

    \node[] at (1,-0.7) {\small visited};
    \node[] at (1,-1.3) {\small not visited};
\end{tikzpicture}
\captionof{figure}{visited nodes if the direction of the extremities are not collinear}\label{fig:non collinear extremities}
\end{minipage}
\end{figure}
\end{proof}

\section{Infinite Grid Exploration with $\Phi=1$} \label{sec:algos}

We now present two algorithms with visibility range one. The former,
Algorithm \algo{1}{Fixed}, uses six robots with three fixed colors. The latter, Algorithm \algo{1}{Modifiable}, uses five robots
with five modifiable colors and achieves exclusiveness.

\subsection{An algorithm using six robots and three fixed colors}
\paragraph{Definition of Algorithm \algo{1}{Fixed}.}

\algoSetRobotColor{F}{green!20!black}
\algoSetRobotColor{L}{blue}
\algoSetRobotColor{B}{red}

\begin{figure}
\centering
  \begin{small}
    \noindent \algoGrid[
\draw[thick] (0,0) -- (4,0) -- (4,-3) -- (0,-3) -- cycle;
\draw[smooth] (0.4,0) to[out=40,in=-170] (1.3,0.5);
\node[] at (2,0.7) {smallest enclosing rectangle};
]{5}{
.,.,.,B,.,
B,.,.,.,.,
.,F,L,.,B,
.,.,B,.,.}{}
  \end{small}
\caption{Initial configuration of  Algorithm \algo{1}{Fixed}.\label{init:algo1}}
\end{figure}
We use the  set of three colors $\Cl=\{L,F,B\}$ to (partially) distinguish robots, \ie, $L$ is the light's color of a robots called \emph{leader}, $F$ is the light's color of a robot called \emph{follower}, and $B$ is the light's color of the four remaining robots, named \emph{beacon} robots.
The initial configuration $I$ of \algo{1}{Fixed} is defined as follows: $I =
\{((-1,0),F), ((0,0),L),$
$((0,-1),B)$, $((2,0),B),$
$((1,2),B), ((-2,1),B)\}$; see Fig.~\ref{init:algo1}.

\algo{1}{Fixed} executes in phases. At the beginning of each phase, we
consider the smallest rectangle, denoted by \SRE, that encloses the
four beacon robots, {\em e.g.}, in the initial configuration $I$
(Fig.~\ref{init:algo1}), the \SRE is drawn with plain lines.  During a
phase, the follower robot $r_F$ explores the borders of the \SRE,
while the leader robot $r_L$ visits the borders of the immediately
smallest inside rectangle.  The group of robots $\{r_L, r_F\}$, called
the {\em moving group}, first moves straightly.  When the leader robot
becomes a neighbor of a beacon robot, the positions of three robots
are {\em adjusted} so that (1) the moving group $\{r_L, r_F\}$ makes a
turn, and (2) the beacon robot moves diagonally in order to expand the
\SRE.  Precisely, at the end of Phase $i$ (and so at the beginning of
Phase $i+1$), both the length and width of \SRE increases by two.

\begin{figure}
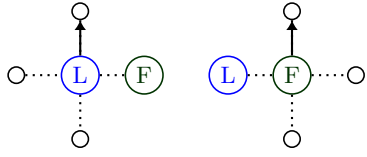

\centering    \algoStep{., ., L, F, .}{up} \quad  \algoStep{., L, F, ., .}{up}
  \caption{$\mathcal{R}_{\mathit{straightH}}$ and $\mathcal{R}_{\mathit{straightT}}$.}\label{fig:rs}
\end{figure}
The rules of $\algo{1}{Fixed}$ are defined in
Figs.~\ref{fig:rs},~\ref{fig:turn}, and~\ref{fig:up}.\footnote{A
  summary of all the rules for $\algo{1}{Fixed}$ is also given in
  Fig.~\ref{fig:rulesalgo1}, page~\pageref{fig:rulesalgo1}.} Some
rules aim at moving the group of robots $\{r_L, r_F\}$ straightly and
the others are used to manage an {\em adjustment}.  In the following,
we detail how $\{r_L, r_F\}$ moves straightly toward a beacon robot,
does a left turn, and how the reached beacon robot moves
diagonally. Recall that the rules below also describe the algorithm
behavior on the equivalent, rotated, local views.

Using Rules of Fig.~\ref{fig:rs}, if we apply $\algo{1}{Fixed}$ to
$\{((i,j),L),((i+1,j),F)\}$,  we obtain
$\{((i,j+1),L),((i+1,j+1),F)\}$, \ie, the two robots go through  the
translation
$\tr{0,1}\left(\left\{((i,j),L),((i+1,j),F)\right\}\right)$.

If we rotate the two robots with angle $\pi/2$, resp. $\pi$ and
$3\pi/2$, then the robots will move to the left, down, and right,
respectively (\ie, each round the moving group undergoes a
translation).  So, the group $\{r_L, r_F\}$ moves in straight line
when isolated. Depending on the relative position of $r_L$ and $r_F$,
the group either moves up, left, down, or right.

\begin{figure}[tb]
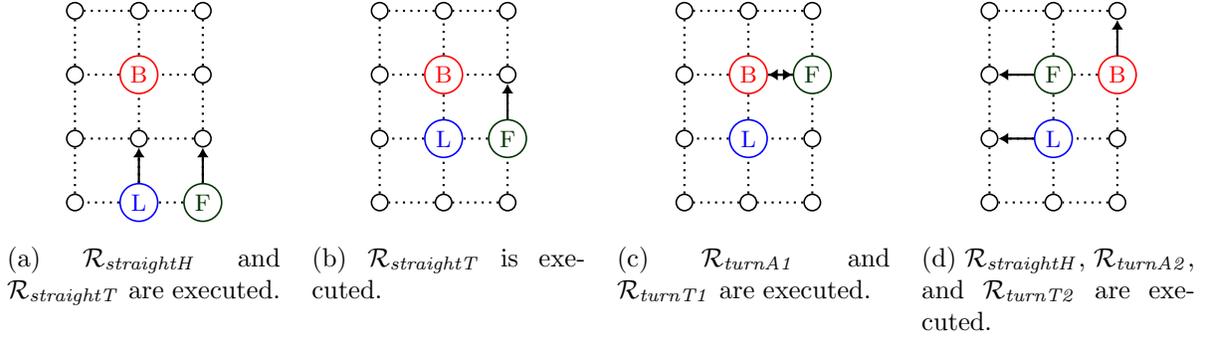

    \centering
    \begin{subfigure}[t]{0.23\textwidth}
\centering\algoGrid{3}{
    ., ., .,
    ., B, .,
    ., ., .,
    ., L, F}{{},up,up}
    \caption{$\mathcal{R}_{\mathit{straightH}}$ and $\mathcal{R}_{\mathit{straightT}}$ are executed.}\label{fig:algo1-global1}
\end{subfigure} \hfill
\begin{subfigure}[t]{0.23\textwidth}
\centering\algoGrid{3}{
., ., .,
., B, .,
., L, F,
., ., .}{{},{},up}
    \caption{$\mathcal{R}_{\mathit{straightT}}$ is executed.}\label{fig:algo1-global2}
\end{subfigure} \hfill
\begin{subfigure}[t]{0.23\textwidth}
\centering\algoGrid{3}{
., ., .,
., B, F,
., L, .,
., ., .}{right,left}
    \caption{$\mathcal{R}_{\mathit{turnA1}}$ and $\mathcal{R}_{\mathit{turnT1}}$ are executed.}\label{fig:algo1-global3}
\end{subfigure} \hfill
\begin{subfigure}[t]{0.23\textwidth}
\centering\algoGrid{3}{
., ., .,
., F, B,
., L, .,
., ., .}{left,up,left}
    \caption{$\mathcal{R}_{\mathit{straightH}}$, $\mathcal{R}_{\mathit{turnA2}}$, and $\mathcal{R}_{\mathit{turnT2}}$ are executed.}\label{fig:algo1-global4}
\end{subfigure}
    \caption{The globally oriented views of the robots performing a turn. In every illustrations, we assume that the nodes around what is visible are empty.}\label{fig:algo1-global}
\end{figure}

\begin{figure}
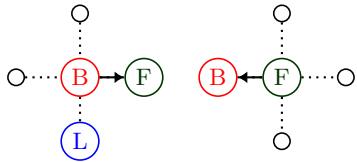

\centering \algoStep{., ., B, F, L}{right}\quad \algoStep{., B, F, ., .}{left}
\caption{$\mathcal{R}_{\mathit{turnA1}}$ and
$\mathcal{R}_{\mathit{turnT1}}$.}\label{fig:turn}
\end{figure}

Before giving the rules for the adjustments and in order to explain clearly
how our algorithm works, we show in Fig.~\ref{fig:algo1-global} the
global configurations that occur when the moving group reaches the
upper right beacon robot.
In the first round, the moving group is translated as previously explained. In
the next two rounds, we can remark that the beacon robot moves one node
to the right and one node up. Concurrently, the moving group $\{r_L,r_F\}$ turns
left to reach a configuration from which it moves in straight line
toward the left.

\begin{figure}
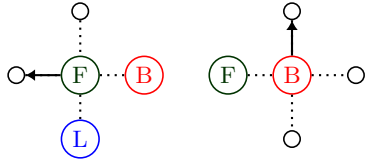

\centering \algoStep{., ., F, B, L}{left} \quad \algoStep{., F, B, ., .}{up}
\caption{$\mathcal{R}_{\mathit{turnA2}}$ and
      $\mathcal{R}_{\mathit{turnT2}}$.}\label{fig:up}
\end{figure}
In more details, for the second round, there is no rule when $r_L$ sees a beacon robot, thus, when it happens $r_L$ stops and $r_F$ continues to move up one more time. For the third round, according to the rules of Fig.~\ref{fig:turn},
when $r_F$ only sees the beacon robot, it moves towards it, and when the beacon sees both $r_F$ and $r_L$, it moves toward $r_F$, so that they exchange their positions, while  $r_L$ stays idle.
Finally, the beacon robot makes a last
move up, and the moving group moves away from the beacon, according to the rules of Fig.~\ref{fig:up}.
With those rules, and with $M=\{((i,j),L), ((i+1,j),F)\}$,
$X=\{((i,j+1),B)\}$, we can see that by applying $\algo{1}{Fixed}$ three times starting from  $\{M,X\}$ we obtain $\{((i-1,j),L), ((i-1,j+1),F),
((i+1,j+2),B)\}$, {\em i.e.}, $\{\rho(M),\;\tr{1,1}(X)\}$, where $\rho$ is the rotation centered at $(i-0.5, j-0.5)$ of angle $\pi/2$.

\begin{theorem}
    Algorithm $\algo{1}{Fixed}$ solves the IGE problem using six robots and fixed colors having common chirality and a visibility range of one.
\end{theorem}

\begin{figure}
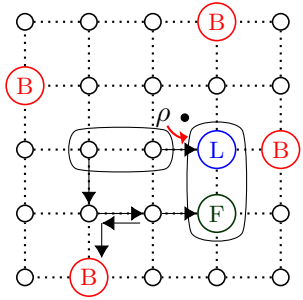

   \centering \algoGrid{5}{
.,.,.,B,.,
B,.,.,.,.,
.,.,.,L,B,
.,.,.,F,.,
.,B,.,.,.}{}[%
\draw[-{Latex[width=2mm]}] (N1-2) -- (N1-3);
\draw[-{Latex[width=2mm]}] (N1-3) -- (N2-3);
\draw[-{Latex[width=2mm]}] (N2-3) -- (N3-3);
\draw[-{Latex[width=2mm]}] (N2-2) -- (N3-2);
\draw[-{Latex[width=2mm]}] (1.8,-3.15) -- (1.2,-3.15);
\draw[-{Latex[width=2mm]}] (1.2,-3.15) -- (1.2,-3.7);
\draw plot [smooth cycle, tension=0.4] coordinates {
(0.8,-1.7) (0.8,-2.3) (2.2,-2.3) (2.2,-1.7)};
\draw plot [smooth cycle, tension=0.4] coordinates {
(2.6,-1.7) (2.6,-3.3) (3.4,-3.3) (3.4,-1.7)};
\node[draw=black, circle, fill, inner sep=1pt,label={180:$\rho$}] at (2.5,-1.5) {};
\draw[thick,red,-{Latex[width=2mm,length=1mm]}] ([shift=(205:0.3cm)]2.5,-1.5) arc (205:280:0.3cm);
    ]
\caption{The globally oriented view  after three rounds  from $C^0$.}\label{fig:global}
\end{figure}
\begin{proof}
  We denote the initial configuration $I = C^0 = \{M^0, C_0^0, C_1^0,
  C_2^0, C_3^0\}$, where $M^0 =\{((-1,0),F),((0,0),L)\}$,
  $C_0^0=\{((0,-1),B)\}$, $C_1^0=\{((2,0),B)\}$,
  $C_2^0=\{((1,2),B)\}$, and $C_3^0=\{((-2,1),B)\}$.

We define the configuration $C^i=\{M^i,
C_0^i, C_1^i, C_2^i, C_3^i\}$ in Phase $i$, where $ M^i =
\tr{-i,-i}(M^0)$, $C_0^i = \tr{-i,-i}(C_0^0)$, $ C_1^i =
\tr{i,-i}(C_1^0)$, $ C_2^i = \tr{i,i}(C_2^0)$, and $ C_3^i =
\tr{-i,i}(C_3^0)$.
We now prove that starting with a configuration $C^i$, the configuration $C^{i+1}$ is eventually reached. Since the initial configuration of
our algorithm is $C^0$, this implies that every configuration $C^i$,
for every $i \geq 0$, is gradually reached.  By doing so, the leader
robot visits all edges of growing rectangles.  Consider the first
configuration $C_i$ of Phase $i$. In $C^i$, the distance between $r_L$
and the beacon robot on its right is $2i+2$.  Indeed, starting from
$C^i$, the robot $r_L$ starts from $(-i,-i)$ and that beacon robot
starts from $(i+2,-i)$.  By executing the algorithm, we can remark
(see Fig.~\ref{fig:global}) that after three rounds (1) the
configuration is $\{\rho(M^i), C^{i+1}_0,C_1^i, C_2^i, C_3^i\}$ (where
$\rho$ is the rotation with center $(0.5; 0.5)$ of angle $\pi/2$) and
(2) $r_L$ is at distance $2i+1$ from the bottom down beacon. From that
point, the moving group $\{r_L, r_F\}$ starts moving one node to the
right at each round (due to the first two rules) until robot $r_L$
sees a beacon robot $r$ in $C_1^i$; this event occurs at round $3+2i$,
\ie, three plus the number of empty nodes between $r_L$ and $r$. After
three more rounds, the moving group performs a left turn again and
bottom right beacon robot is translated by a vector $(1,-1)$.

Thus, at round  $3+2i+3$, the configuration is $\{\tr{2i,0}(\rho^2(M^i)), C^{i+1}_0,C_1^{i+1}, C_2^i, C_3^i\}$. After $2i+3$ more rounds, the moving group reaches the top right beacon robot, and performs another left turn. So at round $3+2(2i+3)$ the configuration is
$\{\tr{2i,2i}(\rho^3(M^i)), C^{i+1}_0,C_1^{i+1}, C_2^{i+1}, C_3^i\}$. Similarly, at round $3+3(2i+3)+1$ the configuration is $\{\tr{-1,2i}(\rho^4(M^i)), C^{i+1}_0,C_1^{i+1}, C_2^{i+1}, C_3^{i+1}\}$. We observe that the moving group  $\{r_L, r_F\}$ required one extra round (as compared to other beacon robots) to reach the beacon robot in $C_3^i$.

Then, after $2i+1$ more rounds, the group of robots $\{r_L, r_F\}$ moves $2i+1$ nodes down
to reach the bottom left beacon robot again, so that, at round $(3+3(2i+3)+1)+2i+1$, the configuration is $\{\tr{-1,-1}(\rho^4(M^i)),$ $ C^{i+1}_0, C_1^{i+1}, C_2^{i+1}, C_3^i\}=C^{i+1}$.

Recursively, if the robots start from configuration $C^0$, they reach configuration $C^i$ in finite time, for any $i\geq 0$. Also, the nodes $V_i$ visited by $r_L$ between Phase $i$ and Phase $i+1$ contains the edges of the rectangle $\left\{\tr{-i,-i}(-1,0), \tr{i,-i}(1,0), \tr{i,i}(1,1), \tr{-i,i}(-1,1)\right\}$; see Fig.~\ref{fig:visit2}.
Since $\bigcup_{i\geq0} V_i = \Z\times\Z$, our algorithm solves the
infinite grid exploration problem.
\end{proof}

\setFigureTextSize{\tiny}
\setFigureScale{0.5}
\tikzstyle{algonode}=[circle,draw=black,semithick, minimum width=1.0em, minimum height=1.0em, inner sep=0.1em]
\setFigureTextSize{\tiny}

\tikzstyle{algonode-small}=[draw=black, circle, semithick, minimum width=0em, minimum height=0em, inner sep=0.00em]

\begin{figure}[htpb]
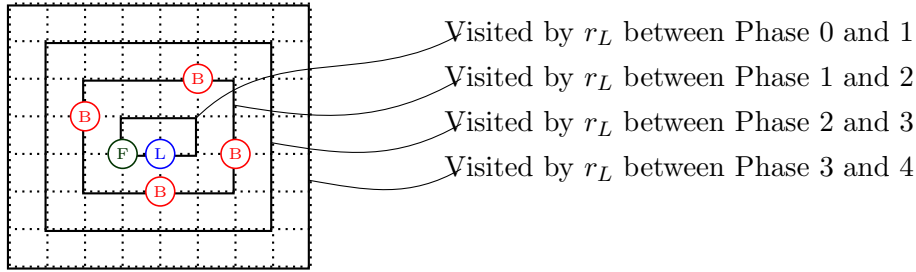

\centering
    \algoGrid[
\def\figoffset{0.05}
\foreach \i in {-1,...,2}
{
    \draw[thick] (2-\figoffset-\i,-5-\figoffset-\i)
     -- (6-\figoffset+\i,-5-\figoffset-\i)
     -- (6-\figoffset+\i,-2-\figoffset+\i)
     -- (2-\figoffset-\i,-2-\figoffset+\i)
     -- cycle;
}
\draw[smooth] (4.95,-3.05) to[out=45,in=-150] (12,-0.8);
\node[] at (17.8,-0.8) {Visited by $r_L$ between Phase 0 and 1};
\draw[smooth] (5.95,-2.7) to[out=-10,in=-150] (12,-2);
\node[] at (17.8,-2) {Visited by $r_L$ between Phase 1 and 2};
\draw[smooth] (6.95,-3.7) to[out=-10,in=-150] (12,-3.2);
\node[] at (17.8,-3.2) {Visited by $r_L$ between Phase 2 and 3};
\draw[smooth] (7.95,-4.7) to[out=-10,in=-150] (12,-4.4);
\node[] at (17.8,-4.4) {Visited by $r_L$ between Phase 3 and 4};
]{9}{
.,.,.,.,.,.,.,.,.,
.,.,.,.,.,.,.,.,.,
.,.,.,.,.,B,.,.,.,
.,.,B,.,.,.,.,.,.,
.,.,.,F,L,.,B,.,.,
.,.,.,.,B,.,.,.,.,
.,.,.,.,.,.,.,.,.,
.,.,.,.,.,.,.,.,.}{}
\caption{Visited grid after three phases for $\algo{1}{fixed}$.}\label{fig:visit2}
\end{figure}

\subsection{An algorithm using five robots and five modifiable colors}

\resetAlgoNodeStyle{}
\setFigureScale{0.85}
\setVisibilityRange{1}

Algorithm \algo{1}{Modifiable} solves the exclusive IGE
using a minimum number of robots. Here, to use one less robot, the
moving group of two robots moves along a triangle, delimited by three beacon robots, instead of a
rectangle like in the previous algorithm.  Except for the shape of the
growing polygonal, the principles are similar to the previous
algorithm.  Notice that we require modifiable colors to allow the
moving group to follow a diagonal.

The fact that the rules are well-defined and unambiguous has been
checked by computer, along with all the transformations that occur
when the robots are in a given configuration. For instance, the fact
that, after a given number of rounds, each beacon at the corner of the triangle has
been translated is verifiable by executing the algorithm in our
complementary material~\cite{quentin_bramas_2019_2625730}.

\algoSetRobotColor{P}{purple!20!black}
\algoSetRobotColor{G}{green!50!black}
\algoSetRobotColor{Y}{yellow!70!black}
\algoSetRobotColor{R}{red!80!black}
\algoSetRobotColor{B}{blue!70!red}

\begin{figure}
\centering
\algoGrid[
\draw[thick] (0,-1) -- (3,-4) -- (2,0) -- cycle;
]{4}{
.,.,R,.,
G,.,.,.,
.,B,.,.,
.,Y,.,.,
.,.,.,Y}{}
\caption{$I$ for \algo{1}{Modifiable}.}\label{fig:init55}
\end{figure}
The set of colors is $\Cl=\{R,Y,G,B,P\}$. Notice that, to reduce the
number of used colors, the meaning of each color changes according to
the stage of the exploration, {\em i.e.}, along the exploration they
are used for different purposes. The initial configuration $I$ is
given in Fig.~\ref{fig:init55}. The three beacon robots are at the
corner of the growing triangle respectively hold light's colors $Y$,
$G$, and $R$. The principle of the algorithm is as follows: starting
from the initial configuration $I$ and using the diagonal movements
described in Fig.~\ref{fig:algo2-sequence-diagonal-move}, the moving
group, composed of the two robots initially with lights colored $B$
and $Y$, goes to the bottom beacon robot $Y$.  The color of the light
of the robot in the moving group initially colored $Y$ alternates at
each move between $Y$ and $P$, while the light of the robot initially
colored $B$ has a constant color.  Robots in the group alternatively
move horizontally and vertically (when one moves horizontally, the
other moves vertically) according to the lights' colors of the group,
either $\{B,Y\}$ or $\{B,P\}$.  After the turn at the bottom beacon
robot, described in Fig.~\ref{fig:algo2-sequence-turn-bottom}, the
lights of the moving group are now colored $G$ and $B$ and the group
moves with fixed colors in the exact same way as in the
previous algorithm, until reaching the third beacon robot. Precisely,
they move up towards the top right beacon robot, turns left, and then
moves straight to the left towards the third beacon robot, following
the rules of the previous algorithm given in
Figs.~\ref{fig:rs}-\ref{fig:turn}. Upon reaching the third beacon
robot, the robots perform a turn following the sequence described in
Fig.~\ref{fig:algo2-sequence-turn-top-left}.  After the turn at the
top left beacon robot, the lights of the moving group have again
colors $B$ and $Y$ and again moves in diagonal. All rules are given in
Fig.~\ref{fig:algomodif}.

\begin{figure}[tb]
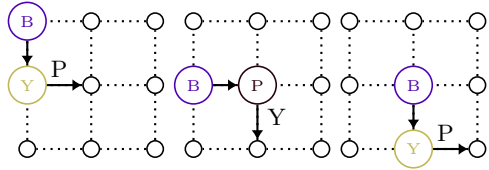
\centering
    \noindent\algoGrid{3}{
B,.,.,
Y,.,.,
.,.,.}{down,right}[
\node at (0.5,-0.75) {\footnotesize P};
]
\noindent\algoGrid{3}{
.,.,.,
B,P,.,
.,.,.}{right,down}[
\node at (1.3,-1.5) {\footnotesize Y};
]
\noindent\algoGrid{3}{
.,.,.,
.,B,.,
.,Y,.}{down,right}[
\node at (1.5,-1.75) {\footnotesize P};
]
\caption{Sequence of configurations that corresponds to the diagonal move. The letter written near each arrow defines the new color of the lights of the moving robot.}\label{fig:algo2-sequence-diagonal-move}
\end{figure}

\begin{figure}[tbh]\centering
    \noindent\algoGrid{3}{
B,.,.,
Y,Y,.,
.,.,.,
.,.,.}{down,right,down}[
\node at (0.5,-0.75) {\footnotesize P};
]
\noindent\algoGrid{3}{
.,.,.,
B,P,.,
.,Y,.,
.,.,.}{right,down,down}[
\node at (1.3,-1.5) {\footnotesize Y};
\node at (1.3,-2.5) {\footnotesize P};
]
\noindent\algoGrid{3}{
.,.,.,
.,B,.,
.,Y,.,
.,P,.}{down,left,right}[
\node at (0.5,-1.75) {\footnotesize G};
\node at (1.5,-2.75) {\footnotesize Y};
]
\noindent\algoGrid{3}{
.,.,.,
.,.,.,
G,B,.,
.,.,Y}{up,up}
\caption{Sequence of configurations that corresponds to the turn at the bottom beacon robot.}\label{fig:algo2-sequence-turn-bottom}
\end{figure}

\begin{figure}[h]\centering
    \noindent\algoGrid{4}{
.,.,.,B,
.,.,G,G,
.,.,.,.}{left,left,left}[
\node at (2.6,-0.75) {\footnotesize Y};
]
\noindent\algoGrid{4}{
.,.,B,.,
.,G,Y,.,
.,.,.,.}{down,left,left}[
\node at (1.6,-0.75) {\footnotesize P};
]
\noindent\algoGrid{4}{
.,.,.,.,
G,P,B,.,
.,.,.,.}{up,down,left}[
\node at (1.25,-1.5) {\footnotesize Y};
]
\noindent\algoGrid{4}{
G,.,.,.,
.,B,.,.,
.,Y,.,.}{{},down,right}[
\node at (1.5,-1.75) {\footnotesize P};
]
\caption{Sequence of configurations that corresponds to the left turn at the top left beacon robot.}\label{fig:algo2-sequence-turn-top-left}
\end{figure}

\resetAlgoNodeStyle{}
\setFigureScale{0.7}
\setVisibilityRange{1}

\begin{figure}[H]

\algoSetRobotColor{P}{purple!20!black}
\algoSetRobotColor{G}{green!50!black}
\algoSetRobotColor{Y}{yellow!70!black}
\algoSetRobotColor{R}{red!80!black}
\algoSetRobotColor{B}{blue!70!red}

For the first moves, we use the same rules as the previous algorithm, allowing two robots to move in straight line toward a beacon, turn left, and move in straight line towards the second beacon:\\
\setFigureScale{0.75}%
\algoStep{., ., G, B, .}{up}
\algoStep{., G, B, ., .}{up}
\algoStep{., R, B, ., .}{left}
\algoStep{., ., R, B, G}{right}
\algoStep{., B, R, ., .}{up}
\algoStep{., ., B, R, G}{left}\\
Rules for the second turn:\\
\algoStep{B, G, G, ., .}{left}[Y]
\algoStep{., ., G, G, .}{left}
\algoStep{B, G, Y, ., .}{left}[P]
\algoStep{., ., B, ., Y}{down}
\algoStep{., ., G, P, .}{up}
\algoStep{., G, P, B, .}{down}[Y]\\
Rules for the diagonal move:\\
\algoStep{B, ., Y, ., .}{right}[P]
\algoStep{., B, P, ., .}{down}[Y]
\algoStep{., ., B, P, .}{right}\\
For the last turn:\\
\algoStep{B, ., Y, Y, .}{right}[P]
\algoStep{., Y, Y, ., .}{down}
\algoStep{., B, P, ., Y}{down}[Y]
\algoStep{P, ., Y, ., .}{down}[P]
\algoStep{Y, ., P, ., .}{right}[Y]
\algoStep{B, ., Y, ., P}{left}[G]
\algoStep{., ., G, Y, .}{left}
\caption{Rules for Algorithm \algo{1}{Modifiable}.}\label{fig:algomodif}
\end{figure}


\bigskip

\begin{theorem}
    Algorithm $\algo{1}{Modifiable}$ solves the IGE problem using five robots, five modifiable colors, and a visibility range of one.
\end{theorem}
\begin{proof}
  First, the fact that using \algo{1}{Modifiable}, robots never
  simultaneously cross the same edge is direct from the definition of
  the rules. In the following, we also show that no two robots ever
  simultaneously occupy the same node and every node is eventually
  visited.  To that goal, we define a sequence of configurations
  corresponding to phases of our algorithm.  We assume that the robot
  with light colored $G$ is initially at position $(0,0)$. So, the
  initial configuration $I$ is defined as follows: $I= C^0= \{M^0,
  C_0^0, C_1^0, C_2^0, C_3^0\}$, where $M^0 =
  \{((1,-1),B),((1,-2),Y)\}$, $C_0^0=\{((0,0),G)\}$,
  $C_1^0=\{((3,-3),Y)\}$, $C_2^0=\{((2,1),R)\}$. The configuration at
  Phase $i$, denoted $C^i = \{M^i, C_0^i, C_1^i, C_2^i\}$, with $i>0$,
  is obtained by several translations, {\em i.e.}, $ M^i =
  \tr{-2i,i}(M^0)$, $C_0^i = \tr{-2i,i}(C_0^0)$, $C_1^i =
  \tr{i,-2i}(C_1^0)$, and $C_2^i = \tr{i,i}(C_2^0)$.

Starting from Configuration $C^i$, the moving group $M^i$ moves
diagonally towards the robot in $C_1^i$ and reaches the first
configuration of the sequence described in
Fig.~\ref{fig:algo2-sequence-turn-bottom}. After the turn, the robot
in $C_1^i$ is translated by a vector $(1,-2)$. Then, the moving group
moves up towards the top-right beacon robot. Next, the moving group turns
left, and the robot in $C_2^i$ is translated by a vector
$(1,1)$. Finally, the moving group reaches the top-left beacon robot and forms
the first configuration of the sequence described in
Fig.~\ref{fig:algo2-sequence-turn-top-left}. After the execution of
this sequence, the robot in $C_0^i$ is translated by vector $(-2,1)$
and the moving group is translated by the same vector with respect to
its position at the beginning of the phase.  At this point, the
configuration is exactly $C^{i+1}$.

Recursively, the execution reaches $C^i$, for all $i\geq 0$, in finite
time. Between Phase $i$ and $i+1$, the nodes located at the edges of
the triangle $\left\{\tr{-2i,i}(0,0), \tr{i,-2i}(2, -2), \tr{i,i}(2,0)\right\}$ are visited; see
Fig.~\ref{fig:visit1}.
\end{proof}

\setFigureScale{0.5}
\begin{figure}[tb]
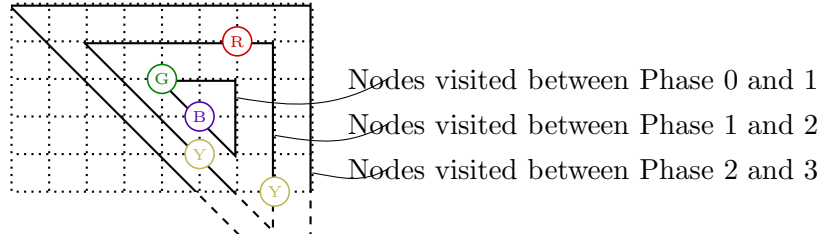

\tikzstyle{algonode}=[circle,draw=black,semithick, minimum width=1.0em, minimum height=1.0em, inner sep=0.1em]
\setFigureTextSize{\tiny}

\tikzstyle{algonode-small}=[draw=black, circle, semithick, minimum width=0em, minimum height=0em, inner sep=0.00em]

\begin{center}
    \algoGrid[
\def\figoffset{0.05}
\begin{scope}
\clip (0,0) rectangle (10,-5);
\foreach \i in {0,...,2}
{
    \draw[thick] (4-\figoffset-2*\i,-2-\figoffset+\i) -- (6-\figoffset+\i,-4-\figoffset-2*\i) -- (6-\figoffset+\i,-2-\figoffset+\i) -- cycle;
}
\end{scope}
\begin{scope}
\clip (0,-5) rectangle (10,-6.1);
\foreach \i in {0,...,2}
{
    \draw[thick,dashed] (4-\figoffset-2*\i,-2-\figoffset+\i) -- (6-\figoffset+\i,-4-\figoffset-2*\i) -- (6-\figoffset+\i,-2-\figoffset+\i) -- cycle;
}
\end{scope}
\draw[smooth] (6,-2.5) to[out=-10,in=-150] (10,-2);
\node[] at (15.2,-2) {Nodes visited between Phase 0 and 1};
\draw[smooth] (7,-3.5) to[out=-10,in=-150] (10,-3.2);
\node[] at (15.2,-3.2) {Nodes visited between Phase 1 and 2};
\draw[smooth] (8,-4.5) to[out=-10,in=-150] (10,-4.4);
\node[] at (15.2,-4.4) {Nodes visited between Phase 2 and 3};
]{9}{
.,.,.,.,.,.,.,.,.,
.,.,.,.,.,.,R,.,.,
.,.,.,.,G,.,.,.,.,
.,.,.,.,.,B,.,.,.,
.,.,.,.,.,Y,.,.,.,
.,.,.,.,.,.,.,Y,.}{}
\end{center}
\caption{Visited grid after three phases for $\algo{1}{Modifiable}$.}\label{fig:visit1}
\end{figure}

\section{Infinite Grid Exploration with $\Phi=2$ and no lights}\label{sec:visibility2}

\algoSetRobotColor{R}{black}

\setFigureScale{0.7}
  \begin{figure}[tb]
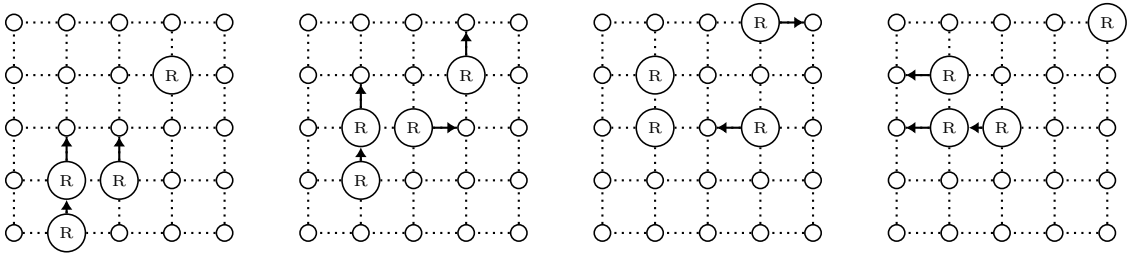

\hfill  \algoGrid{5}{
.,.,.,.,.,
.,.,.,R,.,
.,.,.,.,.,
.,R,R,.,.,
.,R,.,.,.}{{},up,up,up}\hfill
    \algoGrid{5}{
.,.,.,.,.,
.,.,.,R,.,
.,R,R,.,.,
.,R,.,.,.,
.,.,.,.,.}{up,up,right,up}\hfill
    \algoGrid{5}{
.,.,.,R,.,
.,R,.,.,.,
.,R,.,R,.,
.,.,.,.,.,
.,.,.,.,.}{right,{},{},left}\hfill
    \algoGrid{5}{
.,.,.,.,R,
.,R,.,.,.,
.,R,R,.,.,
.,.,.,.,.,
.,.,.,.,.}{{},left,left,left}\hfill
\caption{Sequence of configurations that corresponds to a left turn.}\label{fig:turnleft}
  \end{figure}

\setFigureTextSize{\footnotesize}
\setVisibilityRange{2}

We now describe Algorithm \algo{2}{no lights} which solves the
exclusive IGE problem assuming visibility range two, yet without using
any color (or equivalently, using the same non-modifiable color for all robots), \ie, using
anonymous oblivious robots.  One can observe that when the visibility
range is two (or more) the obstructed visibility can impact the local
view of a robot as a robot at distance one can hide a robot behind it
at distance two. So, the rules of \algo{2}{no lights} should not
depend on the states of the nodes that are hidden by a robot. To make
it clear, those nodes will be crossed out in the illustrations of our
rules.  The principle of our algorithm is similar to the first two
algorithms.  We still proceed by phases. In Phase $i$ ($i \geq 0$), a
moving group, this time of three robots, traverses the edges of a
square of length $2i$; see Fig.~\ref{fig:visit3}.  The three {\em
  moving} robots are always placed in such a way that exactly one of
them, the {\em leader}, has one robot of the group on its horizontal
axis and the other on its vertical axis. Again, the two non-leader
robots of the group are called the {\em followers}. Notice however
that the leadership changes during a phase. Finally, as previously,
the non-members of the moving group are called the {\em beacon}
robots.

\begin{figure}
\begin{center}
        \algoGrid{6}{
.,R,.,.,.,.,
.,.,.,.,.,R,
.,.,.,R,R,.,
.,.,.,R,.,.,
.,.,.,.,.,.,
R,.,.,.,.,.,
.,.,.,.,.,R}{}
\end{center}
\caption{$I$ for \algo{2}{no lights}.}\label{fig:ICnl}
\end{figure}

 The overall idea is that the moving group moves
 straightly according to the relative positions of its members until a
 follower detects a beacon at distance two. Then, an adjustment is
 performed in two rounds to push away the beacon and to make the
 moving group turn left.

The initial configuration is given in Fig.~\ref{fig:ICnl} and the rules are given in Figs.~\ref{fig:line},~\ref{fig:firstr} and~\ref{fig:secondr}.\footnote{A summary of all the rules for $\algo{2}{no lights}$ is also given in Fig.~\ref{fig:rulesalgo3}, page~\pageref{fig:rulesalgo3}.} During Phase $i$, the visited square is actually the one of length $2i$ whose
center is the initial position of the bottom follower.

For the movements along a straight line, the moving group forms a right angle. Each of the three moving robots
sees the others, can determine its position in the group, and knows
the current direction to follow. The rules of Fig.~\ref{fig:line}
manage the movements of the moving group along a straight line.

\begin{figure}
\centering
\algoStep{., ., ., ., ., ., R, R, ., ., R, ., .}{up}
\algoStep{., ., R, R, ., ., R, ., ., ., ., ., .}{up}
\algoStep{., ., ., ., ., R, R, ., ., R, ., ., .}{up}
\caption{Moving on a straight line for \algo{2}{no lights}.}\label{fig:line}
\end{figure}

One can observe that the two last rules are distinguishable by the
robots thanks to their common chirality. Indeed, in one local view,
the central robot sees a robot $r$ at distance one and a robot $r'$ on
the right of $r$. In the other local view, $r'$ is on the left of
$r$. Then, a adjustment is done in two rounds: In the first round
(Fig.~\ref{fig:firstr}), a beacon robot sees a follower in diagonal,
and moves up. Simultaneously, the follower sees the beacon and moves
towards the nodes on the right of the beacon robot.  In the second
round (Fig.~\ref{fig:secondr}), the beacon robot moves away, on the
left of the aforementioned follower it sees at distance two (\ie, on
the right from a global point of view described in
Fig.~\ref{fig:secondr}). Simultaneously, that the follower, which sees
the beacon robot at distance two, catches up with the other robots of
the moving group that are on its left.

\noindent\begin{minipage}{0.4\textwidth}
        \algoStep{., ., ., ., ., ., R, ., ., R, ., ., .}{up}
        \algoStep{., ., ., R, ., R, R, ., ., R, ., ., .}{right}
        
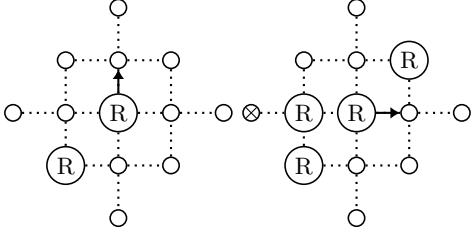
\captionof{figure}{First round.}\label{fig:firstr}
\end{minipage}
\hfill
\begin{minipage}{0.4\textwidth}
        \algoStep{., ., ., ., ., ., R, ., ., ., ., ., R}{right}
        \algoStep{R, ., ., ., R, ., R, ., ., ., ., ., .}{left}
        
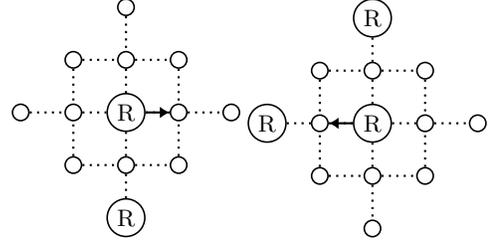
\captionof{figure}{Second round.}\label{fig:secondr}
\end{minipage}

\begin{theorem}
    Algorithm \algo{2}{no lights} solves the exclusive IGE problem using seven robots without lights and a visibility range of two.
\end{theorem}
\begin{proof} 
Again, the fact that using \algo{2}{no lights}, robots never simultaneously cross the same edge is direct from the definition of the rules. In the following, we also show that no two robots ever simultaneously occupy the same node and every node is eventually visited by exhibiting a sequence of
configurations reached by the robots.  $I= C^0 = \{M^0, C_0^0,
C_1^0, C_2^0, C_3^0\}$ is the initial configuration, where $M^0
=\{(3,2),(3,3),$ $(4,3)\}$, $C_0^0=\{(0,0)\}$, $C_1^0=\{(5,-1)\}$,
$C_2^0=\{(5,4)\}$, $C_3^0=\{(1,5)\}$ (with $(0,0)$ being the
bottom-left beacon robot).  $C^i= \{M^i, C_0^i, C_1^i, C_2^i,
C_3^i\}$, where $M^i = \tr{i,i}(M^0)$, $C_0^i = \tr{-i,-i}(C_0^0)$,
$C_1^i = \tr{i,-i}(C_1^0)$, $C_2^i = \tr{i,i}(C_2^0)$, $C_3^i =
\tr{-i,i}(C_3^0)$, is the configuration at Phase $i>0$ obtained from
$C^0$ by several translations.

Starting from Configuration $C^i$, the moving group $M^i$ makes a left turn and the robot in $C_2^i$ is  translated by a vector $(1,1)$. Then, the moving group moves straight on the left towards the top-left beacon, performs a left turn and the robot in $C_3^i$ is  translated by a vector $(-1,1)$.
Next, the moving group moves straight down towards the bottom-left beacon, performs a left turn and the robot in $C_0^i$ is  translated by a vector $(-1,-1)$. Again, the moving group moves straight right towards the bottom-right beacon, performs a left turn and the robot in $C_1^i$ is  translated by a vector $(1,-1)$. Finally, the moving group moves straight up until it reaches the top-right beacon. At this point, the moving group is translated by vector $(1,1)$ with respect to its position at the beginning of the phase, and the configuration is exactly $C^{i+1}$.

Recursively, the execution reaches $C^i$ for all $i\geq 0$ in finite
time (Fig.~\ref{fig:visit2}). Between Phase $i$ and $i+1$, the nodes located at the edges of the square $\left\{\tr{-i,-i}(2,1),\tr{i,-i}(4, 1), \tr{i,i}(4,3), \tr{-i,i}(2,3)\right\}$ are visited by the successive leaders of the moving group,  and the node $(3,2)$ is initially already visited.
\end{proof}

\begin{figure}[H]
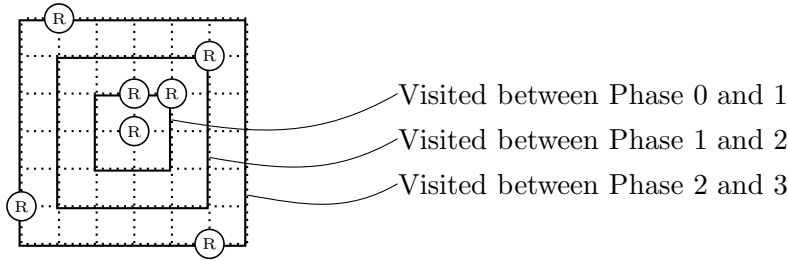

\tikzstyle{algonode}=[circle,draw=black,semithick, minimum width=1.0em, minimum height=1.0em, inner sep=0.1em]

\tikzstyle{algonode-small}=[draw=black, circle, semithick, minimum width=0em, minimum height=0em, inner sep=0.00em]

\setFigureTextSize{\tiny}
\setFigureScale{0.5}

\noindent\begin{center}
    \algoGrid[
\def\figoffset{0.05}
\foreach \i in {0,...,2}
{
    \draw[thick] (2-\figoffset-\i,-4-\figoffset-\i)
     -- (4-\figoffset+\i,-4-\figoffset-\i)
     -- (4-\figoffset+\i,-2-\figoffset+\i)
     -- (2-\figoffset-\i,-2-\figoffset+\i)
     -- cycle;
}
\draw[smooth] (4,-2.7) to[out=-10,in=-150] (10,-2);
\node[] at (15.2,-2) {Visited between Phase 0 and 1};
\draw[smooth] (5,-3.7) to[out=-10,in=-150] (10,-3.2);
\node[] at (15.2,-3.2) {Visited between Phase 1 and 2};
\draw[smooth] (6,-4.7) to[out=-10,in=-150] (10,-4.4);
\node[] at (15.2,-4.4) {Visited between Phase 2 and 3};
]{7}{
.,R,.,.,.,.,.,
.,.,.,.,.,R,.,
.,.,.,R,R,.,.,
.,.,.,R,.,.,.,
.,.,.,.,.,.,.,
R,.,.,.,.,.,.,
.,.,.,.,.,R,.}{}
\end{center}

\caption{Visited grid after three phases for $\algo{2}{no lights}$.}\label{fig:visit3}
\end{figure}

\section{Conclusion}

We have considered the problem of exploring an infinite discrete
environment, namely an infinite grid-shaped graph, using a small
number of mobile synchronous robots with low computation and
communication capabilities.  In particular, our robots are opaque and
only agree on a common chirality.  We show that using modifiable
lights with few states (actually five), five such robots, with a
visibility range restricted to one, are necessary and sufficient to
solve the (exclusive) Infinite Grid Exploration (IGE) problem. We also
provide two other algorithms that respectively solve (1) the
non-exclusive IGE problem using six robots still with visibility range
one but only three constant colors, (2) and the exclusive IGE problem
using seven oblivious anonymous robots yet assuming visibility range
two.

The immediate perspective of this work is to study the optimality, in
terms of number of robots, when we consider the cases of robots where
lights have constant colors or no color at all. As a matter of facts,
we conjecture that our algorithm for six robots is optimal when lights
are assumed to have fixed colors.
As a longer term perspective, we envision to study the IGE problem in
fully asynchronous settings.

\bibliographystyle{plain}
\bibliography{biblio}

\begin{thebibliography}{10}

\bibitem{ABKS18c}
Ranendu Adhikary, Kaustav Bose, Manash~Kumar Kundu, and Buddhadeb Sau.
\newblock Mutual visibility by asynchronous robots on infinite grid.
\newblock In {\em Algorithms for Sensor Systems - 14th International Symposium
  on Algorithms and Experiments for Wireless Sensor Networks, {ALGOSENSORS}
  2018, Helsinki, Finland, August 23-24, 2018, Revised Selected Papers}, pages
  83--101, 2018.

\bibitem{BBMR08j}
Roberto Baldoni, Fran{\c{c}}ois Bonnet, Alessia Milani, and Michel Raynal.
\newblock Anonymous graph exploration without collision by mobile robots.
\newblock {\em Inf. Process. Lett.}, 109(2):98--103, 2008.

\bibitem{BMPT11c}
Fran\c{c}ois Bonnet, Alessia Milani, Maria Potop-Butucaru, and S\'{e}bastien
  Tixeuil.
\newblock Asynchronous exclusive perpetual grid exploration without sense of
  direction.
\newblock In Antonio~Fern{\'a}ndez Anta, editor, {\em Proceedings of
  International Conference on Principles of Distributed Systems (OPODIS 2011)},
  number 7109 in Lecture Notes in Computer Science (LNCS), pages 251--265,
  Toulouse, France, December 2011. Springer Berlin / Heidelberg.

\bibitem{BAKS19c}
Kaustav Bose, Ranendu Adhikary, Manash~Kumar Kundu, and Buddhadeb Sau.
\newblock Arbitrary pattern formation on infinite grid by asynchronous
  oblivious robots.
\newblock In {\em {WALCOM:} Algorithms and Computation - 13th International
  Conference, {WALCOM} 2019, Guwahati, India, February 27 - March 2, 2019,
  Proceedings}, pages 354--366, 2019.

\bibitem{quentin_bramas_2019_2625730}
Quentin Bramas, Stephane Devismes, and Pascal Lafourcade.
\newblock {Infinite Grid Exploration by Disoriented Robots : Animations}, May
  2019.
\newblock \url{https://doi.org/10.5281/zenodo.2625730}.

\bibitem{brandt2018tight}
Sebastian Brandt, Jara Uitto, and Roger Wattenhofer.
\newblock A tight lower bound for semi-synchronous collaborative grid
  exploration.
\newblock In {\em 32nd International Symposium on Distributed Computing (DISC
  2018)}, volume 121, page~13. Schloss Dagstuhl-Leibniz-Zentrum f{\"u}r
  Informatik, 2018.

\bibitem{Das:2016:AMR:2853249.2853724}
Shantanu Das, Paola Flocchini, Giuseppe Prencipe, Nicola Santoro, and Masafumi
  Yamashita.
\newblock Autonomous mobile robots with lights.
\newblock {\em Theor. Comput. Sci.}, 609(P1):171--184, January 2016.

\bibitem{Devismes:2012:OGE:2413887.2413894}
St{\'e}phane Devismes, Anissa Lamani, Franck Petit, Pascal Raymond, and
  S{\'e}bastien Tixeuil.
\newblock Optimal grid exploration by asynchronous oblivious robots.
\newblock In {\em Proceedings of the 14th International Conference on
  Stabilization, Safety, and Security of Distributed Systems}, SSS'12, pages
  64--76, Berlin, Heidelberg, 2012. Springer-Verlag.

\bibitem{DN16}
Gabriele Di~Stefano and Alfredo Navarra.
\newblock Gathering of oblivious robots on infinite grids with minimum traveled
  distance.
\newblock {\em Information and Computation}, 254, 09 2016.

\bibitem{DP07j}
Yoann Dieudonn{\'{e}} and Franck Petit.
\newblock Circle formation of weak robots and lyndon words.
\newblock {\em Inf. Process. Lett.}, 101(4):156--162, 2007.

\bibitem{DDC17c}
Durjoy Dutta, Tandrima Dey, and Sruti~Gan Chaudhuri.
\newblock Gathering multiple robots in a ring and an infinite grid.
\newblock In {\em Distributed Computing and Internet Technology - 13th
  International Conference, {ICDCIT} 2017, Bhubaneswar, India, January 13-16,
  2017, Proceedings}, pages 15--26, 2017.

\bibitem{10.1007/978-3-540-69507-3_5}
Asaf Efrima and David Peleg.
\newblock Distributed models and algorithms for mobile robot systems.
\newblock In Jan van Leeuwen, Giuseppe~F. Italiano, Wiebe van~der Hoek,
  Christoph Meinel, Harald Sack, and Franti{\v{s}}ek Pl{\'a}{\v{s}}il, editors,
  {\em SOFSEM 2007: Theory and Practice of Computer Science}, pages 70--87,
  Berlin, Heidelberg, 2007. Springer Berlin Heidelberg.

\bibitem{Emek:2015:MAT:2852367.2852607}
Yuval Emek, Tobias Langner, David Stolz, Jara Uitto, and Roger Wattenhofer.
\newblock How many ants does it take to find the food?
\newblock {\em Theor. Comput. Sci.}, 608(P3):255--267, December 2015.

\bibitem{FLOCCHINI2005147}
Paola Flocchini, Giuseppe Prencipe, Nicola Santoro, and Peter Widmayer.
\newblock Gathering of asynchronous robots with limited visibility.
\newblock {\em Theoretical Computer Science}, 337(1):147 -- 168, 2005.

\bibitem{LFCPSV17j}
Giuseppe Antonio~Di Luna, Paola Flocchini, Sruti~Gan Chaudhuri, Federico
  Poloni, Nicola Santoro, and Giovanni Viglietta.
\newblock Mutual visibility by luminous robots without collisions.
\newblock {\em Inf. Comput.}, 254:392--418, 2017.

\bibitem{OD18c}
Fukuhito Ooshita and Ajoy~K. Datta.
\newblock Brief announcement: Feasibility of weak gathering in
  connected-over-time dynamic rings.
\newblock In {\em Stabilization, Safety, and Security of Distributed Systems -
  20th International Symposium, {SSS} 2018, Tokyo, Japan, November 4-7, 2018,
  Proceedings}, pages 393--397, 2018.

\bibitem{Peleg:2005:DCA:2098351.2098353}
David Peleg.
\newblock Distributed coordination algorithms for mobile robot swarms: New
  directions and challenges.
\newblock In {\em Proceedings of the 7th International Conference on
  Distributed Computing}, IWDC'05, pages 1--12, Berlin, Heidelberg, 2005.
  Springer-Verlag.

\bibitem{YSDT11j}
Yan Yang, Samia Souissi, Xavier D{\'{e}}fago, and Makoto Takizawa.
\newblock Fault-tolerant flocking for a group of autonomous mobile robots.
\newblock {\em Journal of Systems and Software}, 84(1):29--36, 2011.

\end{thebibliography}

\clearpage
\appendix

\section{Additional Figures}

\setFigureScale{0.8}
\setVisibilityRange{1}
\algoSetRobotColor{F}{green!20!black}
\algoSetRobotColor{L}{blue}
\algoSetRobotColor{B}{red}

\begin{figure}[h]\centering
$\mathcal{R}_{\mathit{straightH}}$ and $\mathcal{R}_{\mathit{straightT}}$:\\
        \algoStep{., ., L, F, .}{up} \quad  \algoStep{., L, F, ., .}{up}

$\mathcal{R}_{\mathit{turnA1}}$ and
$\mathcal{R}_{\mathit{turnT1}}$:\\
\algoStep{., ., B, F, L}{right}\quad
\algoStep{., B, F, ., .}{left}

 $\mathcal{R}_{\mathit{turnA2}}$ and
      $\mathcal{R}_{\mathit{turnT2}}$:\\
    \algoStep{., ., F, B, L}{left} \quad
\algoStep{., F, B, ., .}{up}
\caption{Rules of Algorithm \algo{1}{fixed}.}\label{fig:rulesalgo1}
\end{figure}

\algoSetRobotColor{R}{black}
\resetAlgoNodeStyle{}
\setFigureScale{0.7}
\setVisibilityRange{2}

\begin{figure}[H]
Moving on a straight line:\\
    \algoStep{., ., ., ., ., ., R, R, ., ., R, ., .}{up}
    \algoStep{., ., R, R, ., ., R, ., ., ., ., ., .}{up}
    \algoStep{., ., ., ., ., R, R, ., ., R, ., ., .}{up}

    First round of a turn:\\
    \algoStep{., ., ., ., ., ., R, ., ., R, ., ., .}{up}
    \algoStep{., ., ., R, ., R, R, ., ., R, ., ., .}{right}

    Second round of a turn:\\
    \algoStep{., ., ., ., ., ., R, ., ., ., ., ., R}{right}
    \algoStep{R, ., ., ., R, ., R, ., ., ., ., ., .}{left}
    \caption{Rules for Algorithm \algo{2}{nocolors}.}\label{fig:rulesalgo3}
\end{figure}

\algoSetRobotColor{R}{black}

\end{document}